\documentclass[review]{elsarticle}

\usepackage{amsmath,amssymb}
\usepackage{natbib}
\usepackage{lineno}
\modulolinenumbers[5]

\usepackage{xcolor}

\usepackage{array}
\newcolumntype{L}[1]{>{\raggedright\let\newline\\\arraybackslash\hspace{0pt}}m{#1}}
\newcolumntype{C}[1]{>{\centering\let\newline\\\arraybackslash\hspace{0pt}}m{#1}}
\newcolumntype{R}[1]{>{\raggedleft\let\newline\\\arraybackslash\hspace{0pt}}m{#1}}

\biboptions{authoryear}

\def\sqr#1#2{{\vcenter{\vbox{\hrule height.#2pt
  \hbox {\vrule width.#2pt height#1pt \kern#1pt
    \vrule width.#2pt}
  \hrule height.#2pt}}}}
\def\square{\mathchoice\sqr56\sqr56\sqr45\sqr34}
\def\done{\rightline{$\square$}}

\newtheorem{theorem}{Theorem}
\newtheorem{remark}{Remark}
\newenvironment{proof}{{\it Proof.}}{\newline \done}

\def\d{\, \mathrm{d}}

\begin{document}

\begin{frontmatter}

\title{Revisiting the random shift approach for testing in spatial statistics}

\author[Budejovice]{Tom\' a\v s Mrkvi\v cka\corref{mycorrespondingauthor}}
\cortext[mycorrespondingauthor]{Corresponding author}
\ead{mrkvicka.toma@gmail.com}

\author[Prague]{Ji\v r\' i Dvo\v r\' ak}
\ead{dvorak@karlin.mff.cuni.cz}

\author[Castellon]{Jonatan A. Gonz\'alez}
\ead{jmonsalv@mat.uji.es}

\author[Castellon]{Jorge Mateu}
\ead{mateu@mat.uji.es}

\address[Budejovice]{Department of Applied Mathematics and Informatics, Faculty of Economics, University of South Bohemia, Studentsk{\'a} 13, 370 05 \v{C}esk\'e Bud\v{e}jovice, Czech Republic}
\address[Prague]{Department of Probability and Mathematical Statistics, Faculty of Mathematics and Physics, Charles University, Sokolovsk\' a 83, 186 75 Prague, Czech Republic}
\address[Castellon]{Department of Mathematics, University Jaume I, Campus Riu Sec, 12071, Castell\'on de la Plana, Castell\'on, Spain.}


\begin{abstract}
{We consider the problem of non-parametric testing of independence of two components of a stationary 
bivariate spatial process. In particular, we revisit the random shift approach that has become a standard method for testing the independent superposition hypothesis in spatial statistics, and it is widely used 
in a plethora of practical applications. However, this method has a problem of li\-be\-ra\-lity caused by breaking the marginal spatial correlation structure due to the toroidal correction. This indeed causes 
that the {assumption} of exchangability, which is essential for the Monte Carlo test to be exact, is not fulfilled.}

{We present a number of permutation strategies and show that the random shift with the variance correction {brings} a suitable improvement {compared to} the torus correction in the random field case. It reduces the liberality and achieves the largest power from all {investigated} variants. To obtain the variance for the variance correction method, several approaches were studied. The best results were achieved}{, for the sample covariance as the test statistics, with the correction factor $1/n$. This corresponds to the asymptotic order of the variance of the test statistics.}

{In the point process case, the problem of deviations from exchangeability is far more complex and we propose an alternative strategy based on the mean cross nearest-neighbor distance and torus correction. It reduces the liberality but achieves slightly lower power than the {usual} cross $K$-function. }
{Therefore we recommend it, when the point patterns are clustered, where the cross $K$-function achieves liberality.} 

\end{abstract}

\begin{keyword}
{Bivariate patterns, Independence, Random field, Random shift, Spatial point patterns, Torus {correction}, Variance correction}

\MSC[2010] 62G10 \sep 60G55 \sep 60G60
\end{keyword}

\end{frontmatter}

\section{Introduction}
In this paper we consider the problem of non-parametric testing of independence of two components of a stationary bivariate spatial process -- a pair of stationary random fields defined on the same domain or a stationary bivariate point process. The most common setting is that only one realization of the bivariate process is observed, rather than independent replicates, and this is the setting that we employ in this paper. Note that in the point process case we mean by the hypothesis of independence the so-called random superposition hypothesis (superposition of two independent populations of points) rather than the so-called random labeling hypothesis (points of a single process are independently assigned to the two groups), see \citet{Diggle2010a}.

If the domain of observation $W$ is rectangular, a popular approach is to wrap the data onto a torus by identifying the opposite edges of $W$ \citep[p.311]{Diggle2010a}. A test statistic $T_0$ is calculated from the observed data, then one component is repeatedly randomly shifted around the torus obtaining a set of values $T_1, \ldots, T_N$ of the test statistic from the shifted data. In this way we obtain an approximation of the distribution of the test statistic under the null hypothesis of independence, without assuming any parametric model for the data. The test of independence is then performed by assessing how typical or extreme the value $T_0$ is w.r.t. the population $T_1, \ldots, T_N$, obtaining a p-value in the classical Monte Carlo fashion.
    
The reasoning behind this approach is that by randomly shifting one component of the data we break the possible dependence structure between the components while keeping the dependence structure within each component intact. Hence the test is conditional on the observed marginal structure. The toroidal correction is applied so that no data is discarded. If this is not an issue, e.g. in case of a very large data set, the minus correction, also called the border correction, can be applied instead, avoiding any issues discussed below.
    
Focusing now on the random field case, the random shift approach was suggested in literature  by \citet{UptonFingleton1985} and popularized by e.g. \citet{DaleFortin2002}. However, it was noted already in \citet{FortinPayette2002} that this procedure can be too liberal (rejecting too often under the null hypothesis) and therefore parametric methods are usually used in this context. It is possible to use e.g. the Moran eigenvalue regression \citep{DrayEtal2006}, spatial cross correlation \citep{Chen2015} or spatial regression \citep{Cressie1993}. Often the Mantel test was used but \citet{GuillotEtal2013} and \citet{LegendreEtal2015} demonstrated that this should not be used in the spatial context. {Furthermore, a fully parametric approach is followed in \citet{Bevilacqua2015} and parametric extensions to image processing and multivariate geostatistics can be found in \citet{Vallejos2008,Vallejos2012} and \citet{Vallejos2015}. However, we focus in this paper on fully non-parametric approaches that can be generally applied to a large body of problems and data sets.}
	
The Monte Carlo test relies on the exchangeability of $(T_0,T_1, \ldots, T_N)$, i.e. the joint distribution of the random vector must not be changed by any permutation of its elements. This property ensures that the test has the required significance level under the null hypothesis.
The liberality of the random shift test is caused by breaking the {marginal} spatial correlation structure due to the toroidal correction.
Since $T_0$ is the only value computed from data \emph{not} having a break in the spatial correlation structure, the vector of test statistics $(T_0, T_1, \ldots, T_N)$ is not exchangeable. $T_0$ is too often considered to be extreme when compared to $T_1, \ldots, T_N$, implying liberality of such a test. For an illustration see Figure~\ref{fig:histograms} (left).
In fact, when the shifted random field is observed in a larger domain so that no toroidal correction is needed, the liberality is not observed. This is illustrated in Figure~\ref{fig:histograms} {(middle)}. 
	
{In this paper we introduce a new resampling strategy based on the random shift and a test statistics computed from the intersection of the window and the shifted window, which is further {standardized to have constant variance}. This resampling strategy is denoted as random shift with variance correction and the reduction of liberality by this strategy is illustrated in Figure~\ref{fig:histograms} {(right)}.}
	
Another non-parametric approach for testing the independence of a pair of random fields is described in \citet{Viladomat2014}. Just as in the random shift approach, one component field is left intact. For the second component field the observed values are randomly permuted among the observation points and then smoothed and scaled so that the smoothed empirical variogram of the permuted data matches the smoothed {empirical}
variogram of the original data. In this way the second-order structure of the random field is presumably matched but there is no hope that the distribution of the random field is unchanged. {Hence this approach can be viewed as a parametric bootstrap, where the bootstrapped random fields are given by the allowed difference of the smoothed variogram of the permuted data and the original variogram.} 

	
\begin{figure}[tb]
    \includegraphics[width=\textwidth]{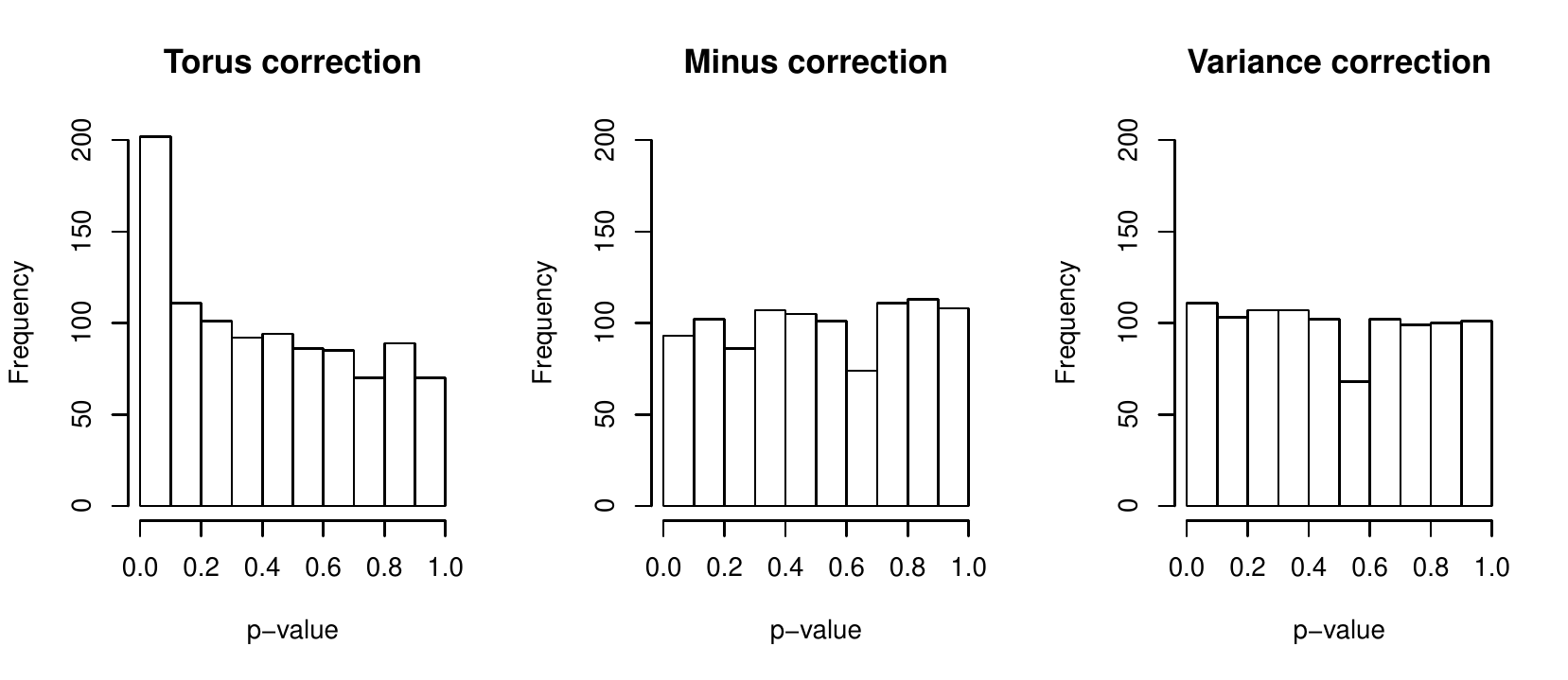}
    \caption{Histograms showing the approximate distribution of p-values for different tests. Based on 1000 replications of a test based on a pair of independent zero-mean unit-variance Gaussian random fields with an exponential correlation function with scale parameter equal to 0.5 (see Section~\ref{sec:simulationsRF} for details).}
    \label{fig:histograms}
\end{figure}

In point process literature the random shift approach was suggested already in \citet{Lotwick1982} who claim on p.410 that ``\ldots it seems intuitively that any statistics calculated after wrapping onto the torus will show less discrepancy from independence, and therefore that spurious significance should not be introduced.'' However, Figure~\ref{fig:histograms} {(left)} and the rest of {this} paper prove the opposite. 
    
The random shift approach has become a standard method for testing the independent superposition hypothesis, see e.g. \citet{Diggle2003,Diggle2010a} or \citet{GrabarnikEtal2011}, but no liberality was mentioned in this context. The random shift approach has seen practical applications e.g. in \citet{SCHLEICHER2011270}, \citet{Fedriani2010} or \citet{Felinks2009}. We show in this paper that also for point processes with long dependence range, the liberality of this test can be surprisingly prominent. To the best of our knowledge there are no parametric methods in the point pattern literature easily available, except of the direct modeling of the bivariate point pattern.

{The paper is structured as follows. Section 2 describes a number of random shift strategies and how they particularize for the random field and point process cases. We prove some theoretical properties for the test statistics based on the sample covariance and for the moment properties of quantities needed to estimate the variance of the cross $K$-function. Simulation experiments come in Section 3, analyzing empirical rejection rates and power of the tests under a variety of scenarios. An application is presented in Section 4. The paper ends with a final discussion.}


\section{Description of methods} 

Assume that we observe two spatial processes $\Phi, \Psi$ in {the same}
observation window $W$. To test the independence between $\Phi$ and $\Psi$ by a permutation method, it would be useful to have independent replicates of the data satisfying the null hypothesis of independence. This is often possible in a parametric setting, using simulations from the fitted model. However, when we restrict ourselves to non-parametric methods, this is not possible. Instead it is important to find a replication strategy which mimics the independent replicates as closely as possible.

Assume that the test is based on the test statistic $T=T(\Phi,\Psi;W)$ and denote by $T_0$ the value computed from the observed data. The test statistic can be scalar, vector or functional, depending on the particular application. As always in the setting of Monte Carlo tests, simulated values $T_1, \ldots, T_N$ are used to approximate the distribution of $T_0$ under the null hypothesis. The outcome of the test is based on how typical or extreme $T_0$ is among the population $T_1, \ldots, T_N$.

\subsection{Random shift strategies}

Below we give an overview of different approaches which share the common framework described above but use different strategies to produce the simulated values $T_1, \ldots, T_N$.

\subsubsection{Torus correction}

The well-established strategy of the \textit{random shift with torus correction} (denoted RS$_{torus}$ in the following) assumes that the window $W$ is rectangular and that $\Psi$ is shifted by a random vector $v$, respecting the toroidal geometry of $W$. Such a toroidal shift makes a crack in the autocorrelation structure of $\Psi$ which causes the liberality of the procedure reported in \citet{FortinPayette2002}. 
Denoting $v_1, \ldots, v_N$ the i.i.d. random vectors used for the shifts and $\left[ \Psi + v_i \right]$ the process $\Psi$ shifted by the $v_i$, respecting the toroidal geometry of $W$, the simulated values are $T_i = T(\Phi,\left[ \Psi + v_i \right];W), i=1,\ldots,N$.

\subsubsection{Minus correction}

The \textit{random shift with minus correction} strategy (denoted RS$_{minus}$ in the following) employs erosion of the window $W$ to $W_c \subset W$ and works only with the restriction of the processes to $W_c$. The distribution of the random shift vectors $v$ is now required  to satisfy $W_c - v \subset W$ almost surely. The simulated values are $T_i = T(\Phi|_{W_c},(\Psi+v_i)|_{W_c};W_c), i=1,\ldots,N$.
In this setting, no cracks in the autocorrelation structure of $\Psi$ are introduced and the desired significance level of the test is achieved. On the other hand, due to the amount of discarded data, the power of the test is small. We remark that in this case the test statistic value computed from the observed data is $T_0 = T(\Phi|_{W_c},(\Psi)|_{W_c};W_c)$. Note that the values $T_i, i=0,\ldots,N,$ are computed from the same amount of data (the observation window $W_c$ is the same for all $i=0,\ldots,N$) and hence are directly comparable.

\subsubsection{Variance correction}

We propose a new strategy which we call the \textit{random shift with variance correction}. It also avoids the cracks in the autocorrelation structure. For the shift vectors $v_1, \ldots, v_N$ denote $W_i = W \cap (W + v_i), i = 1, \ldots, N,$ the smaller window where both the information about $\Phi$ and $(\Psi+v_i)$ is available. The first step consists of producing the simulated values $T_i = T(\Phi|_{W_i},(\Psi+v_i)|_{W_i};W_i), i=1,\ldots,N$. Unlike the previous correction, the values $T_0,\ldots,T_N$ are now computed from different amounts of data (windows $W, W_1,\ldots,W_N$ are different) and hence are not directly comparable. Clearly, the variability of $T_i$ is higher for $W_i$'s with smaller volume.

Thus in the second step, the observed value $T_0$ and the simulated values $T_1,\ldots,T_N,$ are standardized to have zero mean and equal variance, i.e. we subtract the overall mean $\overline{T} = \frac{1}{N+1} \sum_{i=0}^N T_i$ and divide by $\sqrt{\text{var} (T_i)}$:
\begin{align*}
    S_i = \left( T_i - \overline{T} \right) / \sqrt{\text{var} (T_i)}, \quad i = 0, \ldots, N.
\end{align*}
The transformed values are then used to produce the outcome of the test in the classical Monte Carlo fashion.

If we know, at least asymptotically, how the variance of the test statistic depends on the volume of the observation window (or the number of observed points, depending on the particular type of data), a precise formula for the variance can be used.

When the information about the variability of the test statistic is not available, a general approach based on non-parametric kernel regression can be used, where the shift vectors $v_i$ are treated as the explanatory variables. More formally, let $W_0 = W$ and $v_0=o$ be the origin, corresponding to the zero shift. After computing the values $T_0, \ldots, T_N$ based on windows $W_0, W_1, \ldots, W_N$ and shift vectors $v_0, v_1, \ldots, v_N,$ we estimate the variance of $T_i, i=0,\ldots,N,$ using the Nadaraya-Watson estimator {\citep{FanGijbels1996}}:
\begin{align}\label{eq:NW}
    \widehat{\text{var}} (T_i) = & \sum_{k=0}^N \left( T_k - \overline{T} \right)^2 \cdot w_{ik}, \quad i = 0, \ldots, N, \\ w_{ik} = & \frac{K\left( \frac{\|v_i - v_k\|}{h} \right)}{\sum_{j=0}^N K\left( \frac{\|v_i - v_j\|}{h}\right) }, \quad i, k = 0, \ldots, N,
\end{align}
where $K$ is a one-dimensional kernel function, such as the Epanechnikov kernel, and $h>0$ is the bandwidth.


After the standardization, the random variables $S_0, \ldots, S_N$ have the same first and second moments. The test of independence of $\Phi$ and $\Psi$ can be based on assessing how extreme is the value of $S_0$ with respect to the set of values $S_1, \ldots, S_N$. 

{For normally distributed random variables the above standardization {ensures}
that marginals $S_1, \ldots, S_N$ follow the same distribution; in general, the standardized random variables achieve that only approximately. In fact, the exchangeability itself is not satisfied for any of the random shift strategies due to the different correlations between different shifts. Therefore, the property of marginals having the same first and second moments can be considered as a satisfactory property to perform the Monte Carlo test. The empirical significance levels of all strategies are further studied by a simulation study and it is shown that the random shift with variance correction achieves an acceptable significance level in various studied cases.}



We have also studied other replication strategies but with worse results than the random shift {with variance correction}.
These inferior strategies are 1) dividing the window into subwindows, permuting the whole subwindows and combining the information together, 2) dividing the window into subwindows, shifting all subwindows independently inside the whole window, and combining the information together, 3) obtaining the information about the liberality of the random shift by a two stage method similar to that of \citet{Baddeley2017}. 

\subsubsection{Shape of the observation window}

The random shift approach with torus correction assumes a rectangular observation window. While it is possible to extend the approach to windows which are finite unions of (aligned) rectangles, such a procedure would increase the amount of cracks in the autocorrelation structure. Subsequently, it would increase the liberality of the test of independence using random shifts with torus correction.

The approaches using minus correction and variance correction can be applied in case of general (compact) observation windows.

\subsection{Random field case}

Assume now that $\Phi$ and $\Psi$ are random fields observed in $W$. Assume also that $\Psi$ is observed in any location, at least on a fine pixel grid. If this is not the case, kriging can be used to provide an estimate of the unobserved values. Let $X$ be the set of sampling locations where the values of $\Phi$ are observed. $X$ can be either random or non-random, depending on the design of the experiment. We denote by $\Phi(X)$ the vector of values of $\Phi$ observed at sampling points $X$, and similarly for $\Psi(X)$.

For testing the independence {assumption} {for} the two random fields $\Phi$ and $\Psi$, a natural choice of the test statistic is the sample covariance $\text{cov}(\Phi(X),\Psi(X))$. It is known in classical statistics that the variance of the sample covariance, computed from an i.i.d. random sample of size $n$, is of order $1/n$. Below we argue that the same holds also in the case with spatial autocorrelation, see Theorem~\ref{T1}. {Hence ${\text{var}}(T_i) \approx C/n_i$ where $n_i$ is the number of sampling locations in $W_i, i=0,1,\ldots,N,$ and $C$ is a constant; thus setting ${\text{var}}(T_i) \approx 1/n_i$ stabilizes the variance of $S_i$}. We denote {such variance correction} approach RS$_{count}$ in the following.

The approach {using a kernel} estimate of the variance of the sample covariance is denoted here by RS$_{ker}$.

Finally, it is possible to estimate the variograms of the random fields and plug these estimates {into} the exact formula for the variance of the sample covariance (not reported here but easily deduced from the formulas in the proof of Theorem~\ref{T1}). Such method is denoted by RS$_{var}$. The random shift with the sample covariance as the test statistic and torus correction is denoted here by RS$_{torus}$ and with minus correction by RS$_{minus}$.

The following theorem justifies the asymptotic approximation $\text{var} (T_i) \approx 1/n_i$ proposed above in the RS$_{count}$ approach, the test statistic being the sample covariance, {under the assumption of independence of $\Phi$ and $\Psi$}.  

\begin{theorem}\label{T1}
Let $\Phi$ and $\Psi$ be two independent stationary random fields on $\mathbb{R}^d$ with finite second moments and non-negative autocovariance functions $C_\Phi, C_\Psi$. Assume that there is a constant $R>0$ such that $C_\Phi(u-v) = C_\Psi(u-v) = 0$ for $\| u - v \| > R$. Let $X=\{ x_i, i \in \mathbb{N} \}$ be a sequence of observation points such that for each point there are at most $K$ other points within distance $R$. Let $s_n, n = 2, 3, \dots,$ be the sample covariance defined as
\begin{align*}
    s_n = \frac{1}{n-1} \sum_{i=1}^n \left( \Phi(x_i) - \bar{\Phi}_n \right) \left( \Psi(x_i) - \bar{\Psi}_n \right)
\end{align*}
where $\bar{\Phi}_n = \frac{1}{n} \sum_{i=1}^n \Phi(x_i)$ and $\bar{\Psi}_n = \frac{1}{n} \sum_{i=1}^n \Psi(x_i)$ are the sample means. Then $\mathbb{E} s_n = 0$ for each $n \in \mathbb{N}$ and
\begin{align*}
    \mathrm{var}\, s_n = \frac{1}{(n-1)^2} \sum_{i=1}^n \sum_{j=1}^n C_\Phi(x_i-x_j) C_\Psi(x_i-x_j) + o(1/n)
\end{align*}
(in the sense that $\frac{o(1/n)}{1/n} \rightarrow 0$ as $n \rightarrow \infty$) and there are constants $0<C_1 \leq C_2 < \infty$ such that
\begin{align*}
    C_1 \leq \liminf_{n \rightarrow \infty} (n \, \mathrm{var}\, s_n) \leq \limsup_{n \rightarrow \infty} (n \, \mathrm{var}\, s_n) \leq C_2.
\end{align*}
  
\end{theorem}

\begin{proof}
From the {assumptions}
{of stationarity and independence} of $\Phi, \Psi$ it follows that 
\begin{align*}
  \mathbb{E}s_n = \frac{1}{n-1} \sum_{i=1}^n \mathbb{E}\left( \Phi(x_i) - \bar{\Phi}_n \right) \mathbb{E}\left( \Psi(x_i) - \bar{\Psi}_n \right) = 0.
\end{align*}
{We then have} that
\begin{align*}
    \mathrm{var}\, s_n = \mathbb{E} s_n^2 = \frac{1}{(n-1)^2} \sum_{i=1}^n \sum_{j=1}^n \mathbb{E}\left( \Phi(x_i) - \bar{\Phi}_n \right)\left( \Phi(x_j) - \bar{\Phi}_n \right) \mathbb{E}\left( \Psi(x_i) - \bar{\Psi}_n \right)\left( \Psi(x_j) - \bar{\Psi}_n \right).
\end{align*}
Direct calculation shows that for $i, j \in \mathbb{N}$ we have
\begin{align*}
     & \mathbb{E}\left( \Phi(x_i) - \bar{\Phi}_n \right)\left( \Phi(x_j) - \bar{\Phi}_n \right) = \\ & = C_\Phi(x_i - x_j) - \frac{1}{n}\sum_{k=1}^n C_\Phi(x_j - x_k)  - \frac{1}{n}\sum_{k=1}^n C_\Phi(x_i - x_k) + \frac{1}{n^2}\sum_{k=1}^n\sum_{l=1}^n C_\Phi(x_k - x_l)
\end{align*}
and similarly for $\Psi$. Taking the product, double sum and dividing by $(n-1)^2$ we get the leading term
\begin{align*}
    L = \frac{1}{(n-1)^2} \sum_{i=1}^n \sum_{j=1}^n C_\Phi(x_i-x_j) C_\Psi(x_i-x_j)
\end{align*}
while all the other terms are $o(1/n)$. This follows from the non-negativity of $C_\Phi$ and $C_\Psi$ and upper bounds such as $\sum_{k=1}^n C_\Phi(x_i-x_k) \leq (K+1) C_\Phi(0)$. Considering the leading term $L$ we obtain similarly the following upper and lower bounds:
\begin{align*}
    & L \leq \frac{n}{(n-1)^2} (K+1) C_\Phi(0) C_\Psi(0), \\
    & L \geq \frac{n}{(n-1)^2} C_\Phi(0) C_\Psi(0).
\end{align*}
Setting now $C_1 = C_\Phi(0) C_\Psi(0)$ and $C_2 = (K+1) C_\Phi(0) C_\Psi(0)$ completes the proof.
\end{proof}

\begin{remark}
Without more specific assumptions on the positions of the observation points $X=\{ x_i, i \in \mathbb{N} \}$, such as $X=\mathbb{Z}^d$, it is not possible to establish a limit for $n \, \mathrm{var}\, s_n$. On the other hand, the previous theorem ensures that the asymptotic order of the variance is $1/n$. {Furthermore, it is not possible to drop the assumption of bounded support of $C_\Phi, C_\Psi$ without additional assumptions on the properties of $\Phi,\Psi$ such as $\alpha$-mixing.}
\end{remark}

\subsection{Point process case}

Assume now that $\Phi$ and $\Psi$ are stationary point processes observed on $W$. For testing the independence structure of the two point processes $\Phi$ and $\Psi$, a usual choice of the test statistic is the sample cross $K$-function \citep{IllianEtal2008} computed for a number of different ranges.
Roughly speaking, the cross $K$-function carries information about the mean number of points of $\Psi$ up to distance $r$ from an arbitrary point of $\Phi$. The globally corrected ``Ohser-type''  estimator of this function is of the form \citep[p.230]{IllianEtal2008}
\begin{align*}
    \widehat{K}_{12}(r) = \frac{c(r)}{\widehat{\lambda}_\Phi \widehat{\lambda}_\Psi} \sum_{x \in \Phi \cap W} \sum_{y \in \Psi \cap W} \mathbb{I}(\| x-y \| \leq r), \; r > 0,
\end{align*}
where $\widehat{\lambda}_\Phi$ and $\widehat{\lambda}_\Psi$ are {estimated intensities of $\Phi$ and $\Psi$, i.e.}
\begin{align*}
    \widehat{\lambda}_\Phi = \sum_{x \in \Phi \cap W} \frac{1}{|W|},
\end{align*}
and $c(r)$ is an edge correction factor given by 
\begin{align*}
c(r)=\frac{\pi r^2}{\Gamma_W(r)},
\end{align*}
where 
\begin{align*}
  \Gamma_W(r) = \int_W\int_W \mathbb{I}(\| x-y \| \leq r)dxdy = 2\pi\int_0^r t \bar\gamma_W(t) dt,
\end{align*}
with isotropized set covariance function $\bar\gamma_W(t)$ of $W$.

The variance of $\widehat{K}_{12}(r)$ was studied for Poisson processes in \cite{Rajala2018} under the assumption of a fixed number of observed points of $\Phi$ and $\Psi$ in $W$. The paper presents also a formula for general {non-Poisson processes} but without a proof.
We {show} in Theorem \ref{T2} the formula for general {processes}
in the case of random number of points in $W$. The method using this variance for standardization of the sample cross $K$-function is denoted here by RS$_{K,var}$. The approach with kernel estimate of the variance of the sample cross $K$-function is denoted here by RS$_{K,ker}$. The random shift with torus correction applied to the sample cross $K$-function is denoted here by RS$_{K,torus}$ and with minus correction by RS$_{K,minus}$. 

We remark that the sample cross $K$-function is estimated {for} a given number of {arguments}
(50 in our study), which results in the same number of simultaneous Monte Carlo tests. The multiple correction was resolved in our study by the global envelope test \citep{MyllymakiEtal2017, MrkvickaEtal2017}.
Due to this multiple testing correction and the fact that the cross $K$-function summarizes the information from some neighborhood of the observed points (which was not the case for the sample covariance of two random fields) the variance correction tests are conservative and less powerful than the torus correction approach, as shown in the simulation study below.

Therefore, we were looking for a test statistic which would be scalar (avoiding the multiple correction problem which {amplifies} the deviation of the random shift strategy from exchangeability) and which would be less affected by summarizing the information from a neighborhood of the observed points (which {amplifies} the effect of cracks in autocorrelation structure). 

Thus we also study the {expectation of the}
cross {nearest-neighbor distance $D_{12}$, which is the (random) distance from an arbitrary point of} $\Phi$ to the nearest point of $\Psi$. To estimate the expectation $\mathbb{E} D_{12}$ we use the Lebesgue-Stieltjes integral 
\begin{align*}\int r \, \widehat{G}_{12}(\mathrm{d} r)
\end{align*}
where $\widehat{G}_{12}$ is the Kaplan-Meier estimator of the cross 
nearest-neighbor distance distribution function $G_{12}$ \citep{Baddeley1997,IllianEtal2008} which is the distribution function of the random variable $D_{12}$. The random shifts methods applied {to} $\mathbb{E} D_{12}$ will be denoted by RS$_{G,torus}$, RS$_{G,minus}$ and RS$_{G,ker}$.

The following theorem specifies the first- and second-order moment properties of the quantities used to estimate the variance of $\widehat{K}_{12}(r)$ which is to be used in  RS$_{K,var}$.
The theorem can be proved by using the appropriate versions of the Campbell theorem which are given in~\ref{appendix:Campbell}.

\begin{theorem}\label{T2}
Let $\Phi$ and $\Psi$ be two independent stationary point processes on $\mathbb{R}^d$ with intensities $\lambda_1$ and $\lambda_2$ and pair-correlation functions $g_1$ and $g_2$, respectively. Let $W$ be the observation window where both processes are observed. For a given $r>0$ {let $f_r:\mathbb{R}^d \rightarrow [0,\infty)$ be a Borel function and} denote 
\begin{align*}
R = \sum_{x \in \Phi \cap W} \sum_{y \in \Psi \cap W} f_r(x-y).
\end{align*}
Let 
\begin{align*}
\widehat{\lambda}_1 = \sum_{x \in \Phi \cap W} \frac{1}{|W|}, \qquad \widehat{\lambda}_2 = \sum_{y \in \Psi \cap W} \frac{1}{|W|}
\end{align*} 
be the estimated intensities and $S=\widehat{\lambda}_1 \widehat{\lambda}_2$. Then if $f_r(x-y) = \mathbb{I}(\| x-y \| \leq r)$,
\begin{align*}
    \frac{R}{S} = \frac{1}{\widehat{\lambda}_1 \widehat{\lambda}_2} \sum_{x \in \Phi \cap W} \sum_{y \in \Psi \cap W} f_r(x-y)
\end{align*}
is, up to a multiplicative constant, equal to $\widehat{K}_{12}(r)$. It holds that
\begin{align*}
    \mathbb{E} R = \mu_R = & \lambda_1 \lambda_2 \int_{W^2} f_r(u-v) \d u \d v, \\
    \mathbb{E} S = \mu_S = & \lambda_1 \lambda_2,
\end{align*}
\begin{align*}
    \mathrm{var} R = \sigma^2_R = & \lambda_1^2 \lambda_2^2 \int_{W^4}[g_1(u-u') g_2(v-v') - 1] f_r(u-v) f_r(u'-v') \d u \d v \d u' \d v' \\
    & + \lambda_1^2 \lambda_2 \int_{W^3} g_1(u-u') f_r(u-v) f_r(u'-v) \d u \d v \d u' \\
    & + \lambda_1 \lambda_2^2 \int_{W^3} g_2(v-v') f_r(u-v) f_r(u-v') \d u \d v \d v' \\
    & + \lambda_1 \lambda_2 \int_{W^2} f_r(u-v) \d u \d v, \\
    \mathrm{var} S = \sigma^2_S = & \frac{1}{|W|^4} \left( \lambda_1^2 \int_{W^2} g_1(u-v) \d u \d v + \lambda_1 |W| \right)\left( \lambda_2^2 \int_{W^2} g_2(u-v) \d u \d v + \lambda_2 |W| \right) - \lambda_1^2 \lambda_2^2,
\end{align*}
\begin{align*}
    \mathrm{cov}(R,S) = & \frac{\lambda_1^2 \lambda_2^2}{|W|^2} \int_{W^4}[g_1(u-u') g_2(v-v') - 1] f_r(u-v) \d u \d v \d u' \d v' \\
    & + \frac{\lambda_1^2 \lambda_2}{|W|^2} \int_{W^3} g_1(u-u') f_r(u-v) \d u \d v \d u' \\
    & + \frac{\lambda_1 \lambda_2^2}{|W|^2} \int_{W^3} g_2(v-v') f_r(u-v) \d u \d v \d v' \\
    & + \frac{\lambda_1 \lambda_2}{|W|^2} \int_{W^2} f_r(u-v) \d u \d v. \\
\end{align*}
\end{theorem}

\begin{remark}
Using the notation of Theorem~\ref{T2}, the approach of \citet[p.351]{StuartOrd1994}, based on the Taylor expansion of the function $f(R,S) = R/S$, provides an approximation of the variance of the ratio $R/S$:
\begin{align}\label{varratio}
    \text{var} \left( \frac{R}{S} \right) \approx \left(\frac{\mu_R}{\mu_S}\right)^2 \left[ \frac{\sigma^2_R}{\mu_R^2} - 2 \frac{\mathrm{cov}(R,S)}{\mu_R \mu_S}  + \frac{\sigma^2_S}{\mu_S^2} \right].
\end{align}
All the required quantities are given in Theorem~\ref{T2} and hence the method RS$_{K,var}$ can be practically used with this approximated variance and plugged-in estimates of the pair-correlation functions $g_1, g_2$ and intensities $\lambda_1, \lambda_2$.
\end{remark}

\section{Simulation experiments} 

\subsection{Random field case}\label{sec:simulationsRF}

For assessing the performance of the different tests of independence of a pair of random fields considered in this paper, we have performed a simulation study taking advantage of the \texttt{R} packages \texttt{spatstat, geoR} and \texttt{RandomFields} and the code accompanying the paper \citet{Viladomat2014}.

The simulation study is designed as follows. The observation window is the unit square $W=[0,1]^2$. The set of sampling points is given by the binomial point process $X$ with 100 points distributed uniformly in $W$. The random fields $\Phi$ and $\Psi$ are stationary centered Gaussian random fields with the isotropic exponential correlation function $c(r) = \exp\{-r/s\}, r \geq 0,$ where $s>0$ is the scale parameter. We consider different values of $s$ from 0.001 (nearly independent observations) to 0.5 (very smooth realizations). 
The variance of $\Phi$ is 1 in all simulation experiments, the variance of $\Psi$ varies in different settings.

Each test based on random shifts uses $N=999$ independent shifts, the \citet{Viladomat2014} approach uses 1000 bootstrap replicates. The RS$_{minus}$ method uses $W_c = [1/3,2/3]^2$ {and random shift vectors with uniform distribution on $[-1/3,1/3]^2$. Other versions of the random shift method use shift vectors with uniform distribution on a disk centered in the origin and having radius 1/2}. The RS$_{ker}$ methods use the Epanechnikov kernel with bandwidths $b_1=0.05$, $b_2=0.1$ or $b_3=0.15$. The RS$_{var}$ method fits the exponential variogram model by a least-squares approach.
Each experiment is repeated 1000 times 
and the empirical rejection rate is recorded.

We first study the possible liberality of the tests under the null hypothesis of independence. {Both $\Phi$ and $\Psi$ have unit variance here.}
{We consider a nominal level equal to 0.05 and compare it to the empirical rejection rates, see Table~\ref{tab:RFs_null}.}
\begin{table}[tp]
\footnotesize
%
%
  \begin{tabular}{| L{1.5cm} ||  C{1.5cm} | C{1.5cm} | C{1.5cm} | C{1.5cm} | C{1.5cm} |} \hline
    Method & RS$_{torus}$ & RS$_{minus}$ & Viladomat & RS$_{count}$ & RS$_{var}$ \\ \hline \hline 
    $s$ = 0.001 & 0.043         & 0.043         & 0.047         & 0.044         & 0.044             \\ \hline
    $s$ = 0.1   & 0.050         & 0.049         & {\bf 0.036}   & 0.049         & 0.038             \\ \hline
    $s$ = 0.2   & {\bf 0.081}   & 0.041         & 0.043         & 0.059         & {\bf 0.031}       \\ \hline
    $s$ = 0.3   & {\bf 0.075}   & 0.039         & {\bf 0.032}   & 0.057         & {\bf 0.022}       \\ \hline
    $s$ = 0.4   & {\bf 0.105}   & 0.048         & {\bf 0.029}   & 0.058         & {\bf 0.031}       \\ \hline
    $s$ = 0.5   & {\bf 0.109}   & 0.037         & 0.037         & {\bf 0.069}   & {\bf 0.022}       \\ \hline
  \end{tabular}
  
  \vspace{0.25cm}
  
  \begin{tabular}{| L{1.5cm} || C{1.5cm} | C{1.5cm} | C{1.5cm} |} \hline
    Method & RS$_{ker}$ & RS$_{ker}$ & RS$_{ker}$ \\ \hline 
    Bandwidth & $b=0.05$ & $b=0.10$ & $b=0.15$ \\ \hline \hline 
    $s$ = 0.001 & 0.042     & {\bf 0.036}          & {\bf 0.032}     \\ \hline
    $s$ = 0.1   & 0.050     & 0.043                & 0.045           \\ \hline
    $s$ = 0.2   & 0.039     & 0.052                & 0.053           \\ \hline
    $s$ = 0.3   & 0.056     & 0.043                & {\bf 0.036}     \\ \hline
    $s$ = 0.4   & 0.051     & 0.045                & {\bf 0.035}     \\ \hline
    $s$ = 0.5   & 0.047     & 0.047                & 0.043           \\ \hline
  \end{tabular}

  \caption{Empirical rejection rates, under the null hypothesis of independence of a pair of random fields, in the simulation study assessing the {achieved}
  significance level of the tests. Columns determine different versions of the test. Rows determine the smoothness of the random fields, from almost independent observations (scale $s= 0.001$) to very smooth random fields ($s = 0.5$). Kernel estimation of variance was performed with different choices of bandwidth $b$. The confidence interval for the rejection rate based on the binomial distribution for $\alpha=0.05$ is [0.037,0.064]. Numbers falling outside this interval are printed in boldface.}
  \label{tab:RFs_null}
\end{table}

We then investigate also the power of the tests considered here under different alternatives. Let $Z_1, Z_2$ be independent, centered, unit-variance Gaussian random fields with the isotropic correlation function $c(r)$ as above. Then $\Phi=Z_1$ and $\Psi=Z_1+\sigma Z_2$ for some $\sigma > 0$. We consider the values of $\sigma$ to be 2, 4 or 6. The tests of independence were performed on the nominal level of 0.05. The empirical rejection rates are given in Table~\ref{tab:RFs_power}.

\begin{table}[tp]
 {\footnotesize 
    \begin{tabular}{| l | l ||  C{1.5cm} | C{1.5cm} | C{1.5cm} | C{1.5cm} | C{1.5cm} |} \hline
    \multicolumn{2}{| l ||}{Method} & RS$_{torus}$ & RS$_{minus}$ & Viladomat & RS$_{count}$ & RS$_{var}$ \\ \hline \hline 
    $\sigma$ = 2 & $s$ = 0.001 & 0.994 & 0.263 & 0.998 & 0.993 & 0.994 \\ \hline
     & $s$ = 0.2 & 0.738 & 0.161 & 0.617 & 0.742 & 0.646 \\ \hline
     & $s$ = 0.5 & 0.621 & 0.122 & 0.457 & 0.607 & 0.428  \\ \hline \hline
     $\sigma$ = 4 & $s$ = 0.001 & 0.722 & 0.114 & 0.740 & 0.711 & 0.721 \\ \hline
     & $s$ = 0.2 & 0.307 & 0.081 & 0.204 & 0.295 & 0.230 \\ \hline
     & $s$ = 0.5 & 0.286 & 0.070 & 0.143 & 0.236 & 0.124 \\ \hline \hline
     $\sigma$ = 6 & $s$ = 0.001 & 0.378 & 0.072 & 0.394 & 0.385 & 0.388 \\ \hline
     & $s$ = 0.2 & 0.186 & 0.058 & 0.100 & 0.161 & 0.121 \\ \hline
     & $s$ = 0.5 & 0.191 & 0.054 & 0.081 & 0.140 & 0.081 \\ \hline
  \end{tabular}
  
  \vspace{0.25cm}
  
  \begin{tabular}{| l | l || C{1.5cm} | C{1.5cm} | C{1.5cm} |} \hline
    \multicolumn{2}{| l ||}{Method} & RS$_{ker}$ & RS$_{ker}$ & RS$_{ker}$ \\ \hline 
    \multicolumn{2}{| l ||}{Bandwidth} & $b=0.05$ & $b=0.10$ & $b=0.15$ \\ \hline \hline 
    $\sigma$ = 2 & $s$ = 0.001 & 0.888 & 0.991 & 0.990 \\ \hline
     & $s$ = 0.2 & 0.094 & 0.346 & 0.560 \\ \hline
     & $s$ = 0.5 & 0.086 & 0.277 & 0.458 \\ \hline \hline
     $\sigma$ = 4 & $s$ = 0.001 & 0.533 & 0.662 & 0.658 \\ \hline
     & $s$ = 0.2 & 0.090 & 0.158 & 0.211 \\ \hline
     & $s$ = 0.5 & 0.074 & 0.146 & 0.182 \\ \hline \hline
     $\sigma$ = 6 & $s$ = 0.001 & 0.295 & 0.338 & 0.337 \\ \hline
     & $s$ = 0.2 & 0.070 & 0.100 & 0.122 \\ \hline
     & $s$ = 0.5 & 0.070 & 0.090 & 0.105 \\ \hline
  \end{tabular}}
  \caption{Empirical rejection rates, under different alternatives, in the tests of independence of a pair of random fields. All tests were performed on the 5\% significance level.}
  \label{tab:RFs_power}
\end{table}

The {results} indicate that the RS$_{torus}$ approach is in fact valid for very rough random fields with nearly independent observations but it gets more and more liberal with the increasing scale parameter. For very smooth random fields the actual significance level can be even higher than twice the nominal level. In the power study this approach achieved the highest empirical rejection rates but this is only a consequence of the fact that the test based on the torus correction rejects inadequately often even under the null hypothesis.

For the RS$_{minus}$ approach, the empirical rejection rates match the nominal significance levels or even indicate slight conservativeness of the test. On the other hand, the power of the test is, in most cases, the smallest among the tests considered here. This is caused by discarding a significant part of the observed data.

The \citet{Viladomat2014} approach appears to be conservative for smooth random fields but matches the nominal significance levels well for the random fields with nearly independent observations. For such random fields, the power of this test is the highest among the tests considered here. 
On the other hand, for smooth random fields this test is outperformed by the random shift test with variance correction.

Concerning the random shift tests with variance correction proposed in this paper, the RS$_{count}$ achieved small liberality but only for very smooth random fields and simultaneously it has high power.
The RS$_{var}$ and RS$_{ker}$ tests appear to be slightly conservative. The trade-off between the bias and variance in kernel estimation method, represented here by the choice of bandwidth, transforms in the RS$_{ker}$ method into balance between the conservativeness and power. The larger the bandwidth, the higher the power and the bigger the conservativeness.

Finally, we remark that all the tests based on random shifts are very fast to compute (with $N=999$ shifts approx. 1 second per realization on a regular desktop computer) while the \citet{Viladomat2014} test is about 20 times more demanding.

\subsection{Point process case}

For assessing the performance of different tests of independence of a pair of point processes considered in this paper we have also performed an additional simulation study, taking advantage of the \texttt{R} packages \texttt{spatstat, RandomFields} and \texttt{GET}.

The simulation study is designed as follows. The observation window is the unit square $W=[0,1]^2$. The following independent models of $\Phi$ and $\Psi$ were used for verifying the significance level:
\begin{itemize}
    \item[S1:] log Gaussian Cox process (LGCP) with mean of the underlying Gaussian random field $\mu=4.5$, variance $\sigma^2 = 1$, and  correlation function being exponential $c(r)=\exp\{ -r/s \}$ with scale $s=0.5$.
    \item[S2:] LGCP with $\mu=4.5, \sigma^2 = 1, s=0.3$, and an exponential correlation function.
    \item[S3:] LGCP with $\mu=4.5, \sigma^2 = 1, s=0.1$, and an exponential correlation function.
    \item[S4:] Poisson process with intensity 150.
    \item[S5:] Strauss process with $\beta=200, \gamma=0.4, R=0.03$.
    \item[S6:] Strauss process with $\beta=350, \gamma=0.4, R=0.05$.
\end{itemize}

The following models were used for the power study:
\begin{itemize}
    \item[P1:] two LGCPs with $\mu=4.5, \sigma^2 = 1$, an exponential correlation function with $s=0.3$, where the realizations of $\Phi$ and $\Psi$ are based on the same realization of the underlying Gaussian random field.
    \item[P2:] same as P1, with $s=0.1$.
    \item[P3:] $\Phi$ is LGCP with $\mu=4.5, \sigma^2 = 1$, an exponential correlation function with $s=0.3$ and $\Psi$ consists of points of $\Phi$ independently jittered by shift vectors with uniform distribution on a disk with radius 0.1.
    \item[P4:] same as P3, with $s=0.1$.
    \item[P5:] $\Phi$ is a Poisson process with intensity 150, and $\Psi$ consists of points of $\Phi$ independently jittered by shift vectors with uniform distribution on a disk with radius 0.1.
    \item[P6:] Strauss process with $\beta=520, \gamma=0.4, R=0.03$ with points randomly independently labeled by ``1'' or ``2''. $\Phi$ is the process of points labeled by ``1'' and $\Psi$ is the process of points labeled by ``2''.
    \item[P7:] same as P6, but the Strauss process has $\beta=650, \gamma=0.7, R=0.05$.
    \item[P8:] cluster model with hard-core process of parent points (Strauss process with $\gamma=0$ and hard-core distance 0.05) and offspring points distributed uniformly on a disk with radius 0.04. Mean number of offspring per parent point is 5. {Parent points are randomly independently labeled ``1'' or ``2'' and the offsprings inherit their parent's label.}
    \item[P9:] same as P8, but with offspring distributed uniformly on a disk with radius 0.06.
\end{itemize}

All models were chosen so that the intensity is approximately 150 for both $\Phi$ and $\Psi$. Otherwise the design is the same as in the random field case: 1000 independent replications, 999 random shifts, {shift vectors with uniform distribution on a disk centered in the origin and having radius 1/2, except for the minus correction where shift vectors with uniform distribution on $[-1/3,1/3]^2$ are used instead.}

The rejection rates of the test with 5\% nominal level, based on the sample cross $K(r)$ function evaluated in 50 equidistantly spaced points in the interval $r \in [0,0.15]$, are reported in Table~\ref{tab:K12}.
The rejection rates of the test with 5\% nominal level, based on $\mathbb{E} D_{12}$, are reported in Table~\ref{tab:EG}.

\begin{table}[tp]
	\centering
	\footnotesize
	 \begin{tabular}{| c | c | c | c | c | c | c |} \hline
	   Model   &$RS_{{K,torus}}$	&$RS_{{K,minus}}$	&$RS_{{K,ker}}(0.05)$	&$RS_{{K,ker}}(0.10)$   &$RS_{{K,ker}}(0.15)$ &RS$_{K,var}$ 	\\ \hline
			S1 & \bf 0.094 & \bf 0.031 & \bf 0.025 & \bf 0.014 & \bf 0.012 & \bf 0.022  \\ \hline
			S2 & \bf 0.068 & \bf 0.035 & \bf 0.034 & \bf 0.020 & \bf 0.015 &     0.039  \\ \hline
			S3 &     0.043 &     0.040 & \bf 0.029 & \bf 0.012 & \bf 0.004 &     0.057  \\ \hline
			S4 &     0.042 &     0.039 & \bf 0.036 & \bf 0.012 & \bf 0.007 &     0.047  \\ \hline
			S5 &     0.051 &     0.053 & \bf 0.035 & \bf 0.014 & \bf 0.004 &     $--$   \\ \hline
			S6 &     0.049 &     0.041 & \bf 0.032 & \bf 0.016 & \bf 0.006 &     $--$   \\ \hline \hline
			P1 &     0.841 &     0.074 &     0.014 &     0.006 &     0.019 &     0.411  \\ \hline
			P2 &     0.968 &     0.115 &     0.016 &     0.020 &     0.058 &     0.634  \\ \hline \hline
			P3 &     0.996 &     0.123 &     0.004 &     0.005 &     0.022 &     0.546  \\ \hline
			P4 &     0.998 &     0.158 &     0.012 &     0.002 &     0.021 &     0.635  \\ \hline
			P5 &     0.970 &     0.112 &     0.008 &     0.003 &     0.010 &     0.764  \\ \hline \hline
			P6 &     0.985 &     0.110 &     0.045 &     0.308 &     0.574 &     $--$   \\ \hline
			P7 &     0.768 &     0.084 &     0.022 &     0.067 &     0.151 &     $--$   \\ \hline \hline
			P8 &     0.457 &     0.033 &     0.023 &     0.039 &     0.055 &     $--$   \\ \hline
			P9 &     0.272 &     0.043 &     0.019 &     0.027 &     0.018 &     $--$   \\ \hline
		\end{tabular}
		\caption{Empirical rejection rates  of various tests of independence with sample cross $K(r)$ function. The confidence interval (based on binomial distribution) for nominal level $\alpha=0.05$ is [0.037,0.064]. Numbers falling outside these intervals in cases studying significance level are printed in boldface.}
		\label{tab:K12}
\end{table}

\begin{table}[tp]
	\centering
	\footnotesize
	 \begin{tabular}{| c | c | c | c | c | c |} \hline
							Model	&$RS_{{G,torus}}$	&$RS_{{G,minus}}$	&$RS_{\text{G,ker}}(0.05)$	&$RS_{{G,ker}}(0.10)$   &$RS_{{G,ker}}(0.15)$ 	\\ \hline
			S1 & 0.054 & \bf 0.035 & \bf 0.036 & 0.040  & \bf 0.026 \\ \hline
			S2 & 0.050 & 0.062  & 0.049 & \bf 0.028 & \bf 0.025 \\ \hline
			S3 & 0.048 & \bf 0.073 & 0.043 & \bf 0.028 & \bf 0.019 \\ \hline
			S4 & 0.057 & 0.044  & 0.043 & 0.037 & \bf 0.026 \\ \hline
			S5 & 0.050 & 0.057  & 0.045 & 0.042 & \bf 0.026 \\ \hline
			S6 & 0.055 & 0.044  & 0.046 & 0.035 & \bf 0.023 \\ \hline
			\hline
			P1 & 0.656 & 0.025 & 0.021 & 0.013 & 0.026 \\ \hline
			P2 & 0.918 & 0.022 & 0.027 & 0.044 & 0.138 \\ \hline \hline
			P3 & 0.911 & 0.014 & 0.009 & 0.003 & 0.007 \\ \hline
			P4 & 0.932 & 0.021 & 0.014 & 0.001 & 0.013 \\ \hline
			P5 & 0.506 & 0.039 & 0.014 & 0.008 & 0.028 \\ \hline \hline
            P6 & 0.906 & 0.134 & 0.167 & 0.601 & 0.743 \\ \hline
			P7 & 0.653 & 0.110 & 0.061 & 0.244 & 0.363 \\ \hline  \hline
			P8 & 0.222 & 0.070 & 0.058 & 0.061 & 0.084 \\ \hline
			P9 & 0.197 & 0.064 & 0.043 & 0.049 & 0.048 \\ \hline
		\end{tabular}
		\caption{Empirical rejection rates of various tests of independence {based on $\mathbb{E} D_{12}$ with 5\% nominal level.}}
		\label{tab:EG}
\end{table}

The point process case is more complex than the random field case since the test statistic accumulates information from certain neighborhood of observed points and hence the variance correction methods do not perform well: the effect of dropping part of information after the shifts is more severe here.
{They are} too conservative and have smaller power than RS$_{torus}$ method.
The effect is more pronounced for cross $K$-function than for {$\mathbb{E} D_{12}$} which is caused by the multiple testing problem. Indeed, the results of RS$_{K,ker}$ computed for a single tested value $K_{12}(0.05)$ (not reported here) are comparable with results of RS$_{G,ker}$. The RS$_{K,minus}$ method is conservative, too, which is probably caused by a long reach of the cross $K$-function and small size of the eroded window.

The RS$_{K,torus}$ method shows the same liberality as in the random field case for clustered {processes}. On the other hand, it does not show liberality for the Poisson process and repulsive processes. The RS$_{G,torus}$ seems to remove the liberality of the torus correction completely. It also seems to have smaller power than RS$_{K,torus}$ which is for clustered processes partly explained by the liberality of RS$_{K,torus}$. 

{The RS$_{K,var}$ method is {by} far the most complicated method in terms of performance and application. This is due to the required knowledge of the true variance for every shifted point pattern. When the true pair-correlation function of each of the two superimposed point patterns is known, the variance can be approximated by Eq. \eqref{varratio}. An estimate of the pair-correlation function could be plugged into the variance estimator, however, it might add some approximation errors. Thus, we only test those cases where the theoretical pair-correlation functions of the superimposed point patterns are known. 
We used high-dimensional Monte Carlo integration with a uniform sample of $160.000$ points within each observation window ({corresponding to} the shifted point patterns). As expected, we have reasonable values for the empirical rejection rates, except for the {largest} scale case (S1) where the test {is}
very conservative. The power, however, is not {very high}
but always better than any other method {except of the torus correction.}}
	
\section{Data examples} 

To show the methodology on real data we chose the tropical tree data sets from Barro Colorado Island plot \citep{Hubbell2005,Condit1998,Hubbell1999}.

\subsection{Random fields}

Four spatial covariates accompanying the tree data were chosen for our analysis. These were soil acidity and nitrogen, phosphorus and potassium content. The plots of these covariates can be seen in Figure~\ref{fig:cov}. We performed the tests of independence of each pair of these random fields, and the resulting $p$-values are given in Table~\ref{tab:RFs_real}. The tests were performed in the window 1000x500 m$^2$ with random shift vectors distributed uniformly on a disk with radius 250 m. For the exact variance method the range 500 m was used for variogram estimation. {Bandwidth for the kernel method was {chosen} to be 25, 50 or 75 meters, respectively.}

\begin{figure}
    \centering
    \includegraphics[scale=0.58]{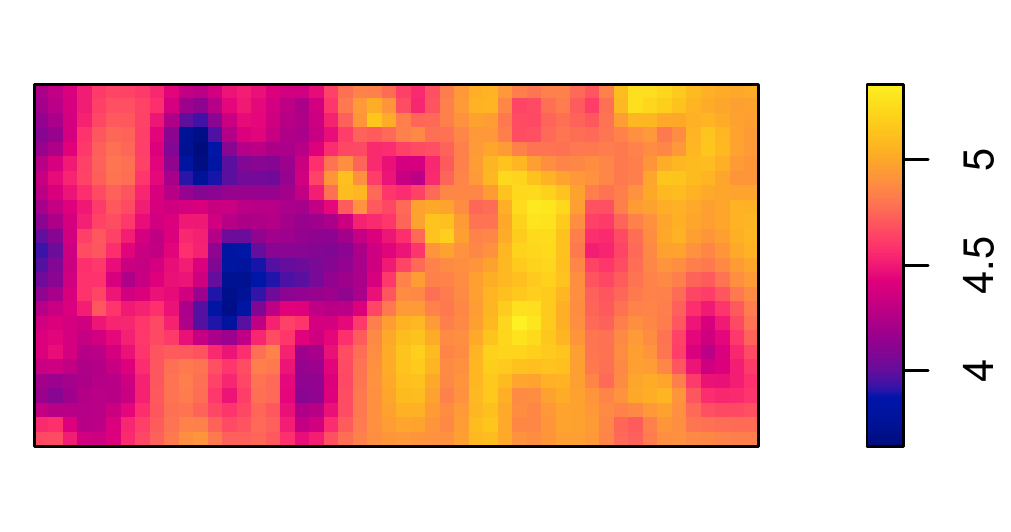}
    \includegraphics[scale=0.58]{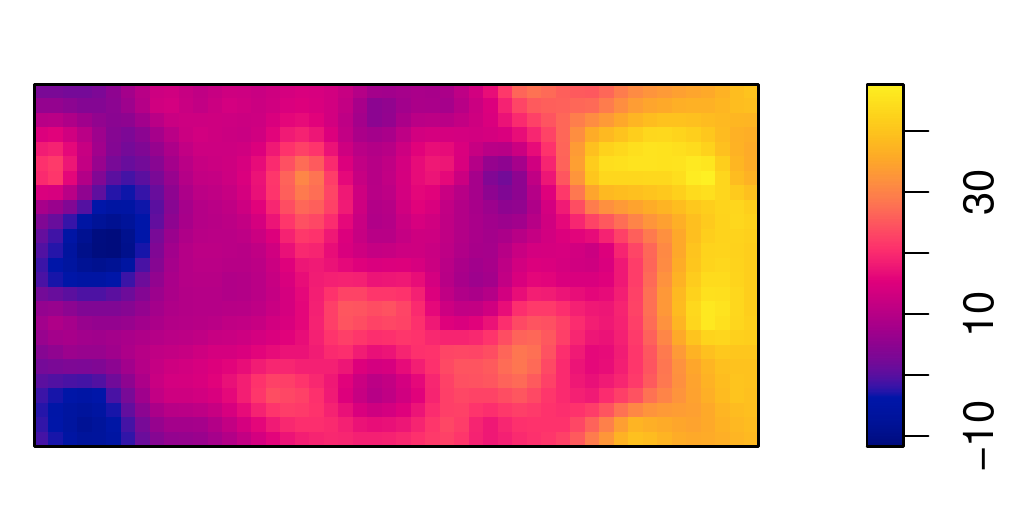}
    \includegraphics[scale=0.58]{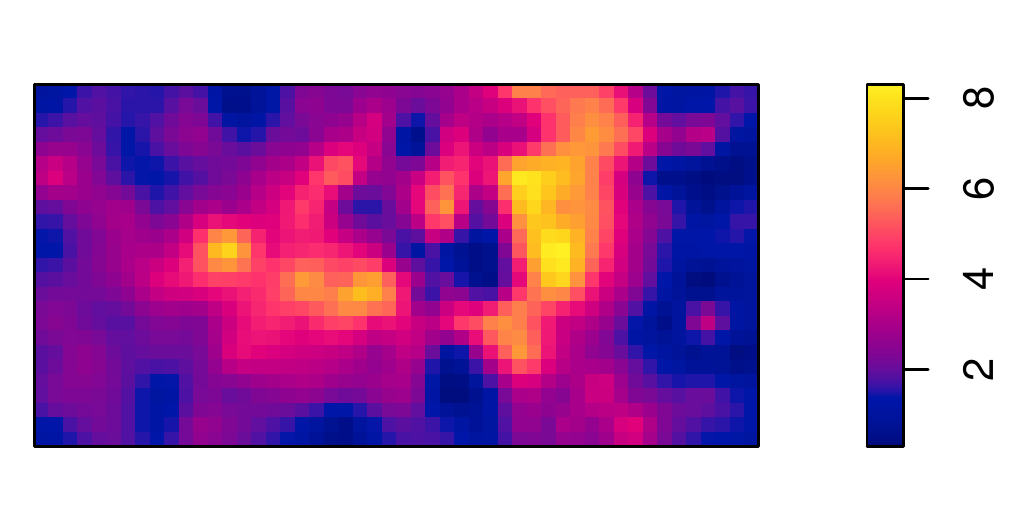}
    \includegraphics[scale=0.58]{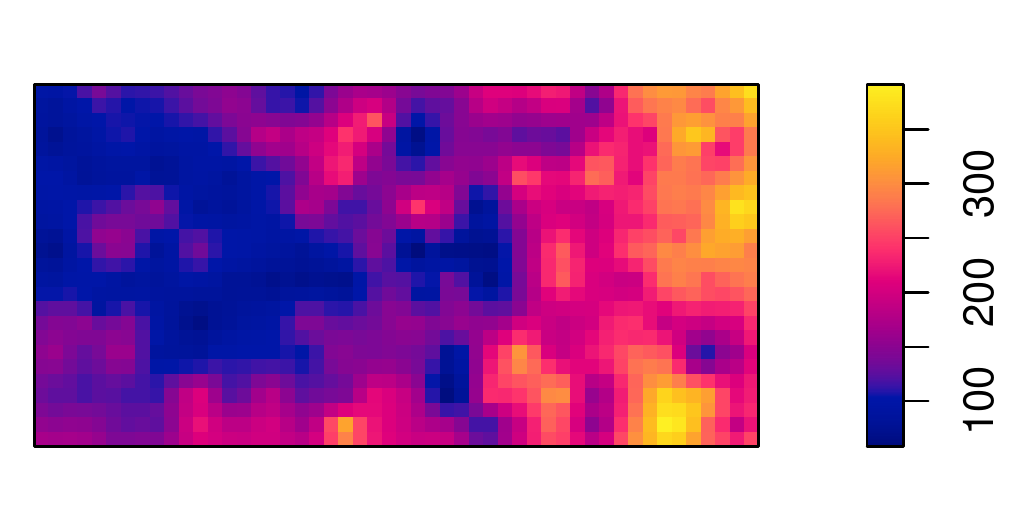}
    \caption{Covariates accompanying the BCI data set. Top row: acidity and nitrogen content, bottom row: phosphorus and potassium content, respectively.}
    \label{fig:cov}
\end{figure}

\begin{table}[tp]
  \centering
   {\footnotesize \begin{tabular}{| l | l | c | c | c | c | c | c | c | c |} \hline
    Method &  RS$_{tor.}$ & RS$_{min.}$ & RS$_{cou.}$ & RS$_{var}$ & $_{ker}(25)$ & $_{ker}(50)$ & $_{ker}(75)$ \\ \hline \hline 
    acidity--nitrogen & 0.018 & 0.408 & 0.578 & 0.708 & 0.882 & 0.890 & 0.868 \\ \hline
    acidity--phosphorus & 0.844 & 0.492 & 0.780 & 0.884 & 0.946 & 0.940 & 0.938 \\ \hline
    acidity--potassium & 0.012 & 0.730 & 0.032 & 0.622 & 0.152 & 0.170 & 0.160 \\ \hline
    nitrogen--phosphorus & 0.412 & 0.070 & 0.406 & 0.764 & 0.646 & 0.698 & 0.702 \\ \hline
    nitrogen--potassium & 0.094 & 0.292 & 0.034 & 0.290 & 0.798 & 0.790 & 0.784 \\ \hline
    phosphorus--potassium & 0.966 & 0.268 & 0.904 & 0.926 & 0.876 & 0.652 & 0.624 \\ \hline

  \end{tabular}}
  \caption{$p$-values of the test of independence of two specified random field serving as the covariates in the BCI data set.}
  \label{tab:RFs_real}
\end{table}

Two $p$-values were found significant at the 0.05 level for the proposed RS$_{count}$ method. Two $p$-values were also found significant by the RS$_{torus}$ method. When applying Bonferroni's correction for multiple testing in the set of 6 tests by a given method, no $p$-value was found significant.

\subsection{Point patterns}

For independence analysis of point patterns two related species (Inga Goldmanii, Inga Sapinoides) were chosen from the BCI data set. Figure~\ref{fig:BCIpatterns} shows the chosen species in the rectangular observation window with sides 500 and 1000 meters long. The total number of observed trees are 313 and 230, respectively.

\begin{figure}
    \centering
    \includegraphics[scale=0.8]{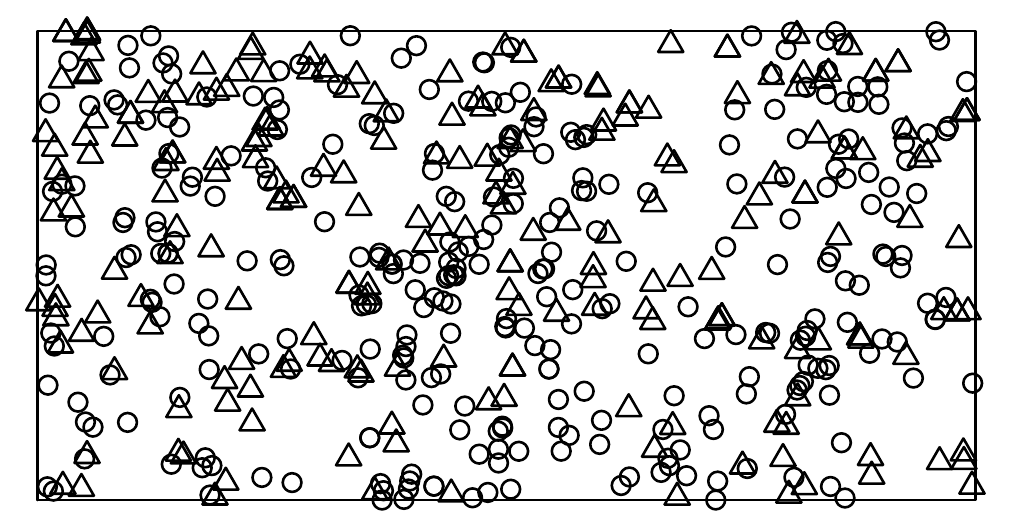}
    \caption{The positions of trees of two species selected from the BCI data set: Inga Goldmanii (circles), Inga Sapinoides (triangles).}
    \label{fig:BCIpatterns}
\end{figure}

Figure~\ref{fig:BCIK} shows the result of the global envelope test performed with the extreme rank length measure of the test of independence based on 50 values of the cross $K$-function in the range 0 to 75 meters and random shifts with torus correction. The random shift vectors were generated uniformly on a disk with radius 250 m. The test produces p-value 0.913. Other variants of the test with the cross $K$-function also did not give significant result.

\begin{figure}
    \centering
    \includegraphics[scale=0.9]{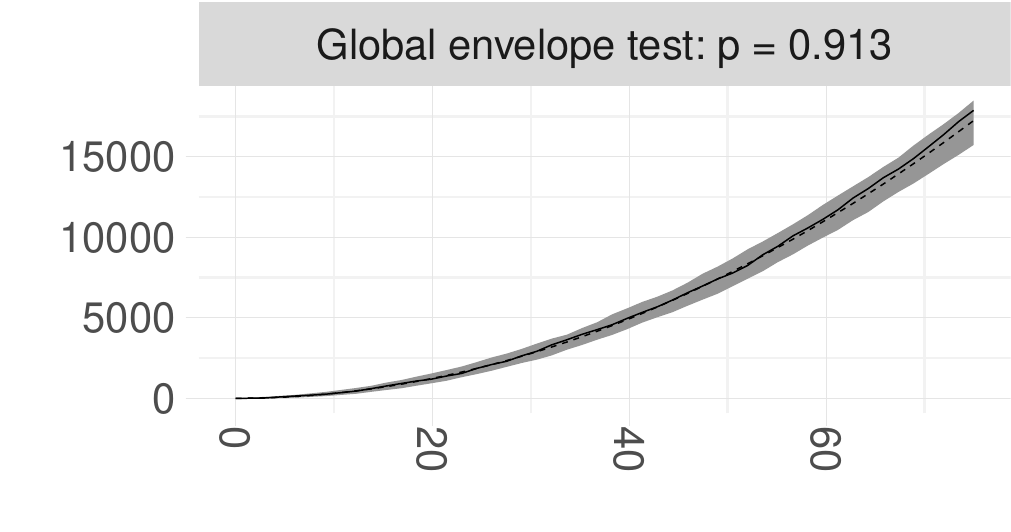}
    \caption{Global envelope test of independence of two point processes {(positions of trees of two species selected from the BCI dataset)}, based on  the cross $K$-function and random shifts with torus correction.}
    \label{fig:BCIK}
\end{figure}
To compare, the test based on $\mathbb{E} D_{12}$ {and torus correction} produces p-value 0.594. Other variants of the test also do not give significant result.

\section{Discussion} 

In this paper we {address the problem of liberality in the random shift with a torus correction permutation strategy}. We studied {a number of} permutation strategies and found out the random shift with the variance correction to be a suitable improvement of the torus correction in the random field case. It reduces the liberality and {achieves the {highest}} power from all variants we studied. It has even higher power than the method of \citet{Viladomat2014} for random fields with some {non-trivial} autocorrelation structure. To obtain the variance for the variance correction method, several approaches were studied. The best results were achieved with the asymptotic order of variance $1/n$ of the sample covariance, i.e. the RS$_{count}$ method. Thus this method can be recommended as the improvement of the torus correction in the random field case, {both for its simplicity and its performance}.

In the point process case we realized that the problem of deviations from exchangeability is more complex than for the random field case; {the variance correction method with kernel smoothing RS$_{K,ker}$ does not perform good enough, and the variance correction method with computed variance RS$_{K,var}$ is too complex and time consuming in order to be generally applied.} 

{However,} we have found another test statistic than the usual cross $K$-function---the mean cross nearest-neighbor distance $\mathbb{E} D_{12}$---which reduces the liberality of the torus correction but achieves slightly lower power than the cross $K$-function. $\mathbb{E} D_{12}$ has two advantages for the random shift strategy: first, it is only a scalar, second, it only considers the {nearest-neighbor} distance, thus it does not extend the dependence in a local test statistic more than it is necessary.
Therefore we can recommend, when the point patterns are repulsive or Poisson, to use the torus correction with the cross $K$-function, whereas in the case of clustered point patterns to use the torus correction with $\mathbb{E} D_{12}$. 

This work was originally motivated by studying the independence in more complex problems of finding a correlation between marks of a point process and covariates. This problem will be addressed in the future. Indeed, the RS$_{ker}$ approach is very general and can be applied to other types of problems as well, especially to the problems where the  information is known at the sampled point and is not accumulated from a neighborhood. The problem of marks and covariates fits into this setting.

\section*{Acknowledgements}

This work was supported by Grant Agency of Czech Republic (Project No.\ 19-04412S). The BCI forest dynamics research project was founded by S.P. Hubbell and R.B. Foster and is now managed by R. Condit, S. Lao, and R. Perez under the Center for Tropical Forest Science and the Smithsonian Tropical Research in Panama. Numerous organizations have provided funding, principally the U.S. National Science Foundation, and hundreds of field workers have contributed. The Barro Colorado Island soils data set was collected and analysed by J. Dalling, R. John, K. Harms, R. Stallard and J. Yavitt with support from National Science Foundation grants DEB021104, DEB021115, DEB0212284, DEB0212818 and
OISE 0314581, Smithsonian Tropical Research Institute and Center for Tropical Forest Science.

\section*{References}


\appendix
\section{Campbell theorems for independent point processes}\label{appendix:Campbell}

\begin{theorem}\label{thm:campbell1}
Let $\Phi$ and $\Psi$ be two independent stationary point processes on $\mathbb{R}^d$ with intensities $\lambda_1$ and $\lambda_2$, respectively, and let $h_1: \left(\mathbb{R}^d\right)^2 \rightarrow [0;\infty)$ be a measurable function. Then
\begin{align}
    \mathbb{E} \sum_{x \in \Phi} \sum_{y \in \Psi} h_1(x,y) = \lambda_1 \lambda_2 \int_{\left(\mathbb{R}^d\right)^2} h_1(u,v) \d u \d v.
\end{align}
\end{theorem}

\begin{proof}
First we define measures $\alpha_1, \alpha_2$ on $\mathcal{B}(\mathbb{R}^d)$ by $\alpha_1(B) = \lambda_1 |B|, \alpha_2(B) = \lambda_2|B|, B \in \mathcal{B}(\mathbb{R}^d),$ where $|\cdot|$ denotes the Lebesgue measure of appropriate dimension. We further define
\begin{align*}
    \alpha(B_1 \times B_2) = \mathbb{E} \sum_{x \in \Phi} \sum_{y \in \Psi} \mathbb{I}_{B_1}(x) \mathbb{I}_{B_2}(y), \; B_1, B_2 \in \mathcal{B}(\mathbb{R}^d).
\end{align*}
Using the classical Campbell formula and independence of $\Phi,\Psi$ we obtain for $B_1, B_2 \in \mathcal{B}(\mathbb{R}^d)$
\begin{align*}
    \alpha (B_1 \times B_2) = \mathbb{E} \sum_{x \in \Phi} \sum_{y \in \Psi} \mathbb{I}_{B_1}(x) \mathbb{I}_{B_2}(y) = \left( \mathbb{E} \sum_{x \in \Phi} \mathbb{I}_{B_1}(x) \right)\left( \mathbb{E} \sum_{y \in \Psi} \mathbb{I}_{B_2}(y) \right) = \lambda_1 |B_1| \cdot \lambda_2 |B_2|.
\end{align*}
Hence $\alpha$ is the product measure $\alpha_1 \otimes \alpha_2$ and from uniqueness of the product measure we have that $\alpha(C) = \lambda_1 \lambda_2 |C|$ for $C \in \mathcal{B}(\mathbb{R}^{2d})$. From this point we proceed using the standard measure theoretic arguments, as in the proof of the classical Campbell theorem, to obtain the claim.
\end{proof}

\begin{theorem}\label{thm:campbell2}
Let $\Phi$ and $\Psi$ be two independent stationary point processes on $\mathbb{R}^d$ with intensities $\lambda_1$ and $\lambda_2$ and pair correlation functions $g_1$ and $g_2$, respectively. Let $h_2:\left(\mathbb{R}^d\right)^4 \rightarrow [0;\infty)$ be a measurable function. Then
\begin{align}
    \mathbb{E} \sum_{x, x' \in \Phi} \sum_{y, y' \in \Psi} h_2(x,x',y,y') = & \lambda_1^2 \lambda_2^2 \int_{\left(\mathbb{R}^d\right)^4} h_2(u,u',v,v') g_1(u-u') g_2(v-v') \d u \d v \d u' \d v' \\
    & + \lambda_1^2 \lambda_2 \int_{\left(\mathbb{R}^d\right)^3} h_2(u,u',v,v) g_1(u-u') \d u \d v \d u' \nonumber \\
    & + \lambda_1 \lambda_2^2 \int_{\left(\mathbb{R}^d\right)^3} h_2(u,u,v,v') g_2(v-v') \d u \d v \d v' \nonumber \\
    & + \lambda_1 \lambda_2 \int_{\left(\mathbb{R}^d\right)^2} h_2(u,u,v,v) \d u \d v. \nonumber
\end{align}
\end{theorem}

\begin{proof}
Similarly to the proof of Theorem~\ref{thm:campbell1} we write for $B_1,B_2 \in \mathcal{B}(\mathbb{R}^{2d})$
\begin{align*}
    \beta(B_1 \times B_2) = & \mathbb{E}\sum_{x, x' \in \Phi} \sum_{y, y' \in \Psi} \mathbb{I}_{B_1}(x,x') \mathbb{I}_{B_2}(y,y') = \left( \mathbb{E}\sum_{x, x' \in \Phi} \mathbb{I}_{B_1}(x,x') \right)\left( \mathbb{E}\sum_{y, y' \in \Psi} \mathbb{I}_{B_2}(y,y') \right) \\
     = & \left( \int_{\left(\mathbb{R}^d\right)^2} \mathbb{I}_{B_1}(x,x') \lambda_1^2 g_1(x-x') \d x \d x' + \int_{\mathbb{R}^d} \mathbb{I}_{B_1}(x,x) \lambda_1 \d x \right) \\
     & \cdot \left( \int_{\left(\mathbb{R}^d\right)^2} \mathbb{I}_{B_2}(y,y') \lambda_2^2 g_2(y-y') \d y \d y' + \int_{\mathbb{R}^d} \mathbb{I}_{B_2}(y,y) \lambda_2 \d y \right).
\end{align*}
Having established $\beta$ as a product measure, we proceed by standard measure theoretic arguments to finish the proof.
\end{proof}

\end{document}